\documentclass[12pt]{article}
\usepackage{amsmath,amssymb,amsthm}
\usepackage{mathbbol,bbm}
\usepackage{graphicx,float}
\usepackage[final]{showkeys}
\usepackage{bm}

\newcommand{\Cx}{{\mathbb C}}

\newcommand{\idty}{\b 1}
\DeclareMathOperator{\id}{id}

 \DeclareMathOperator{\tr}{Tr}
\newcommand{\<}{\langle}
\renewcommand{\>}{\rangle}
\providecommand{\abs}[1]{\lvert#1\rvert}
\providecommand{\norm}[1]{\lVert#1\rVert}
\DeclareMathOperator*{\loplus}{\mbox{\Large\mbox{$\oplus$}}}

\renewcommand{\b}[1]{\mathbf{#1}}
\newcommand{\bi}[1]{\boldsymbol{#1}}
\renewcommand{\c}[1]{\mathcal{#1}}

\newcommand{\s}[1]{\mathsf{#1}}
\renewcommand{\r}[1]{\mathrm{#1}}

\newtheorem{theorem}{Theorem}
\newtheorem{lemma}{Lemma}
\newtheorem{proposition}{Proposition}

\begin{document}
\setlength{\parskip}{6pt}

\begin{center}
{\LARGE On thermalization in Kitaev's 2D model} \\[12pt]
R.~Alicki$^{\dagger,\sharp}$, M.~Fannes$^\ddagger$ and M.~Horodecki$^{\dagger,\sharp}$
\\[6pt]
$^\dagger$ National Quantum Information Centre of Gda\'nsk, Poland \\[6pt]
$^\sharp$ Institute of Theoretical Physics and Astrophysics \\
University of Gda\'nsk, Poland \\[6pt]
$^\ddagger$ Instituut voor Theoretische Fysica \\
K.U.Leuven, Belgium
\end{center}

\medskip\noindent
\textbf{Abstract}

\noindent 
The thermalization process of the 2D Kitaev model is studied within the Markovian weak coupling approximation. It is shown that its largest relaxation time is bounded from above by a constant independent of the system size and proportional to $\exp(2\Delta/kT)$ where $\Delta$ is an energy gap over the 4-fold degenerate ground state. This means that the 2D Kitaev model is not an example of a memory, neither quantum nor classical.

\section{Introduction}
\label{s1}

The fragility of genuine quantum states in the presence of interaction with an environment, in contrast to those which allow a classical interpretation, represents the main challenge for the large scale implementation of the ideas of quantum information processing.

It is well-known that to improve the stability of quantum information processing the logical qubits should  be implemented in many-particle systems, typically $N$ physical spins per logical qubit. The logical qubits should be stable objects with efficient methods of state preparation, measurements and application of gates. By efficiency we mean certain scaling behaviour, e.g.\ the lifetime of a logical qubit should grow exponentially with $N$ while the number of steps needed to apply a gate and a measurement should at most increase polynomially with $N$.

There exist different sources of noise, some of them of macroscopic origin and often highly system-dependent. They can, at least in principle, be substantially reduced by proper engineering. However, one source of noise is inescapable: the microscopic interactions of the physical spins with thermal particles or excitations of the local environment.

The analogous situation for classical information processing is well-under\-stood. E.g.\ ferromagnetic interactions stabilize bits encoded is small magnetic domains which can then be flipped by external magnetic fields. The existence of a similar mechanism for quantum information is still open to debate.

As a first step, one is interested in designing a stable quantum memory, i.e.\ a $N$- particle system which can support at least a single encoded logical qubit for a long time, preferably growing exponentially with $N$. The most promising theoretical candidates are provided by quantum spin systems which exhibit topological order, such as the 2D \cite{KitaevFT} and 4D Kitaev models \cite{DennisKLP2002-4DKitaev}. The issue of thermal stability of these models has recently attracted much attention (see e.g.\cite{AlickiH2006-drive,Alicki-Fannes-Horodecki,NussinovO2007b,Kay2008-K2D}). The aim of this paper is to rigorously analyze the thermal stability of the 2D Kitaev model. 

The quantum theory of open systems developed more than thirty years ago provides a natural framework for studying stability in the presence of thermal noise.  The particularly simple properties of Kitaev's model allow to apply the most developed tools in the theory of quantum open systems namely Davies's theory of semi-groups of completely positive maps which describes the dynamics of a quantum system weakly interacting with a heat bath in the Markovian approximation. Such evolutions satisfy the quantum detailed balance condition which permits a rather detailed analysis of the relaxation properties of the relevant degrees of freedom.

We first develop some technical tools which provide upper bounds for the relaxation times of observables. Then we apply this to the study two examples: a 1D quantum Ising model, which can be seen as a 1D Kitaev model, and the 2D  Kitaev model, both coupled to a heat bath. We show that the relaxation times of the observables which are supposed to encode the qubits are bounded from above by a constant, independent of the system size. Therefore, neither the 1D nor the 2D model provides a stable quantum memory. Although this result was expected, the most important issue is the mathematical technique which is developed here, it could be applied to analyze other models such as the 4D Kitaev model, as well. The method relies on a Hamiltonian picture of the dissipation where estimating the longest relaxation time is shown to be equivalent to estimate the ground state spectral gap for a very particular quantum spin Hamiltonian. These techniques, together with those developed in our previous paper \cite{Alicki-Fannes-Horodecki} on this subject could be applied to analyze other models.

The paper is organized as follows: in Section~\ref{s2} we briefly remind the origins and the main properties of Davies generators as the quantum analogues of classical Glauber dynamics. Section~\ref{s3} relates the lifetime of information encoded in a relaxing system to the spectral gap of the dissipative generator. In Section~\ref{s4} a Hamiltonian picture of relaxation is obtained. The generator is recast into a master Hamiltonian on Liouville space such that the gap of the generator is recovered as the ground state gap of the master Hamiltonian. Section~\ref{s5} deals with a simple model: the 1D Ising model on a ring which can encode 1 topological qubit in its ground state. A rigorous argument is developed to show that the stability of this model when coupled to a thermal environment does not increase with the system size. Finally, in Section~\ref{s6} Kitaev's 2D model is shown to behave in a similar way as the 1D Ising model with respect to thermal stability.
\section{Davies generators}
\label{s2}

Davies generators provide a simple and realistic description of a system weakly coupled to a thermal environment. The interaction with the outside world is so small that the system doesn't loose its identity and can still be considered on itself. The footprint of the thermal environment is in the reduced dynamics: a Markovian semi-group of unity preserving completely positive maps of a very specific   type with the temperature encoded. The aim of this section is to briefly recall the general form and properties of such maps, henceforth called Davies maps, see~\cite{Davies1974-WCL} for the original rigorous analysis of the weak coupling limit.

A small system, with a finite number of states, is coupled to one or more heat baths at a same inverse temperature $\beta$ leading to a total Hamiltonian
\begin{equation} 
H = H^{\text{sys}} + H^{\text{bath}} + H^{\text{int}} \enskip\text{with}\enskip H^{\text{int}} = \sum_\alpha S_\alpha \otimes f_\alpha,
\end{equation} 
where the $S_\alpha$ are system operators and the $f_\alpha$ bath operators. Both the coupling operators $S_\alpha$ and $f_\alpha$ are assumed to be Hermitian. An important ingredient is the Fourier transform $\hat g_\alpha$ of the auto-correlation function of $f_\alpha$. The function $\hat g_\alpha$ describes the rate at which the coupling is able to transfer energy between the bath and the system. Often a minimal coupling to the bath is chosen, minimal in the sense that the interaction part of the Hamiltonian is as simple as possible but still addresses all energy levels of the system Hamiltonian in order to yield an ergodic reduced dynamics. The necessary and sufficient conditions for ergodicity have been obtained in~\cite{Spohn1977-ergod,Frigerio1978-ergod}
\begin{equation} 
\bigl\{ S_\alpha, H^{\text{sys}} \bigr\}' = \Cx\, \idty, 
\end{equation} 
i.e.\ no system operator apart from the multiples of the identity commutes with all the $S_\alpha$ and $H^{\text{sys}}$.

The weak coupling limit results in a Markovian evolution for the system
\begin{equation}
\frac{dX}{dt} = \c G(X) = i\delta(X) + \c L(X). 
\label{reddyn}
\end{equation}
The generator is a sum of two terms, the first is a usual Liouville-von~Neumann derivation as in standard quantum mechanics while the second is a particular type of Lindblad generator
\begin{align}
\delta(X) 
&= [H^{\text{sys}},X] \\
\c L(X) 
&= \sum_\alpha \sum_{\omega\ge0} \c L_{\alpha\,\omega}(X)
\label{dis} \\
&= \sum_\alpha \sum_{\omega\ge0} \hat g_\alpha(\omega)\, \Bigl\{ \bigl(S_\alpha(\omega)\bigr)^\dagger\, \bigl[X \,,\, S_\alpha(\omega) \bigr] + \bigl[ \bigl( S_\alpha(\omega) \bigr)^\dagger \,,\, X \bigr]\,  S_\alpha(\omega) 
\nonumber \\
&\phantom{=\ } + \r e^{-\beta\omega}\, S_\alpha(\omega)\, \bigl[ X \,,\, \bigl( S_\alpha(\omega) \bigr)^\dagger \bigr] + \r e^{-\beta\omega}\, \bigl[ S_\alpha(\omega) \,,\, X \bigr]\, \bigl( S_\alpha(\omega) \bigr)^\dagger \Bigr\}.
\label{dav}
\end{align}
Here the $S_\alpha(\omega)$ are the Fourier components of $S_\alpha$ as it evolves under the Hamiltonian system evolution
\begin{equation}
\r e^{itH^{\text{sys}}}\, S_\alpha\, \r e^{-itH^{\r{sys}}} = \sum_\omega S_\alpha(\omega)\, \r e^{-i\omega t}, 
\label{salpha}
\end{equation}
where the $\omega$'s are the Bohr frequencies of the system Hamiltonian. The Hermiticity of $S_\alpha$ is equivalent with
\begin{equation} 
\bigl( S_\alpha(\omega) \bigr)^\dagger = S_\alpha(-\omega).
\end{equation} 
Clearly, the temperature of the environment appears in~(\ref{dav}) through the Boltzmann factor. Such generators are called Davies generators in the sequel.

A super-operator $\c L$ as in~(\ref{reddyn}) generates a semi-group of completely positive identity preserving transformations of the system. However, due to its specific form, it enjoys a number of important additional properties \\
$\bullet$\enskip the canonical Gibbs state is stationary
\begin{equation} 
\tr \Bigl( \rho_\beta\, \r e^{t\c G}(X) \Bigr) = \tr \bigl( \rho_\beta\, X \bigr) \enskip\text{with}\enskip \rho_\beta = \frac{\r e^{-\beta H^{\r{sys}}}}{\tr \Bigl( \r e^{-\beta H^{\r{sys}}} \Bigr)}, 
\end{equation} 
$\bullet$\enskip the semi-group is relaxing: any initial state $\rho$ evolves to $\rho_\beta$
\begin{equation} 
\lim_{t\to\infty}\ \tr \Bigl( \rho\, \r e^{t\c G}(X) \Bigr) = \tr \bigl( \rho_\beta\, X \bigr), 
\end{equation} 
$\bullet$\enskip each of the terms $\c L_{\alpha\omega}$, and therefore $\c L$ as well, satisfies the detailed balance condition, often called reversibility
\begin{equation}
[\delta, \c L_{\alpha\omega}] = 0 \enskip\text{and}\enskip \tr \Bigl(\rho_\beta\, Y^\dagger\, \c L_{\alpha\omega}(X) \Bigr) = \tr \Bigl(\rho_\beta\, \bigl(\c L_{\alpha\omega}(Y)\bigr)^\dagger\, X \Bigr). 
\label{db}
\end{equation}

Equation~(\ref{db}) expresses the self-adjointness of $\c L_{\alpha\omega}$ with respect to the Liouville scalar product
\begin{equation} 
\<X \,,\, Y\>_\beta := \tr \rho_\beta\, X^\dagger\, Y. 
\end{equation} 
Writing $\c G = i\delta + \c L$, see~(\ref{reddyn}), is therefore decomposing $\c L$ into a Hermitian and a skew-Hermitian part and, by~(\ref{db}), $\c G$ is normal. The dissipative terms $\c L_{\alpha\omega}$ of the generator, being self-adjoint, are negative definite. This implies that we can always sandwich the spectrum of $\c L$ between spectra of generators where the $\hat g(\omega)$ have been replaced by constants.

The negativity of $\c L_{\alpha\omega}$ becomes manifest if we write
\begin{equation}
\begin{split}
- \bigl\< X \,,\, \c L_{\alpha\omega}(X) \bigr\>_\beta =
&\bigl\< \bigl[ S_\alpha(\omega) \,,\, X \bigr] \,,\, \bigl[ S_\alpha(\omega) \,,\, X \bigr] \bigr\>_\beta \\
&+ \r e^{-\beta\omega}\, \bigl\< \bigl[ S^\dagger_\alpha(\omega) \,,\, X \bigr] \,,\, \bigl[ S^\dagger_\alpha(\omega) \,,\, X \bigr]\bigr\>_\beta.
\end{split}
\label{liopos}
\end{equation}
To show this equality one uses that
\begin{equation} 
\rho_\beta\, S_\alpha(\omega) = \r e^{\beta\omega}\, S_\alpha(\omega)\, \rho_\beta, 
\end{equation} 
which is a direct consequence of~(\ref{salpha}).

\section{Relaxation times and memory}
\label{s3}

Due to ergodicity any initial state of a system whose dynamics is governed by a Davies generator will eventually relax to equilibrium. Information can be encoded by perturbing the equilibrium state of the system and, in order to retrieve this information, one has to single out observables that detect the perturbation of the state. The longer lived these observables, the more efficient the memory. The aim is indeed to increase the useful life of observables and states by devising proper Hamiltonians and encoding procedures. We are in this paper interested in the 2D Kitaev model of size $N$ and especially in the relation between lifetime and system size. We begin by arguing that the main point of interest is the temporal behaviour of auto-correlation functions (for a detailed analysis of relation 
between auto-correlation functions and the fidelity criterion see \cite{Alicki-Fannes-fidelity}).

To information encoded in an initial state $\rho$ we associate a complex observable $A$ through
\begin{equation} 
\rho = \rho_\beta (\idty + A) \enskip\text{with}\enskip \tr \bigl( \rho_\beta A \bigr) = 0. 
\end{equation} 
The observable $A^\dagger$ can then be used to recover this information as
\begin{equation} 
\tr \bigl( \rho\, A^\dagger \bigr) = \tr \bigl( \rho_\beta A\,A^\dagger\bigr) > 0. 
\end{equation} 
This detection, however, disappears in the long run
\begin{equation} 
\lim_{t\to\infty} \tr \Bigl( \rho\, \r e^{t\c G}(A^\dagger) \Bigr) = \tr \bigl( \rho_\beta A^\dagger \bigr) = 0.
\end{equation} 
By the detailed balance property of $\c L$ we have
\begin{equation} 
\tr \Bigl( \rho\, \r e^{t\c G}(A^\dagger) \Bigr) = \tr \Bigl( \rho_\beta A\, \r e^{t\c G}(A^\dagger) \Bigr) = \tr \Bigl( \rho_\beta\, \r e^{t\c L/2}(A)\, \r e^{t(i\delta + \c L/2)}(A^\dagger) \Bigr). 
\end{equation} 
We then apply Schwarz's inequality to obtain
\begin{align} 
\Bigl| \tr \Bigl( \rho_\beta\, A\, \r e^{t\c G}(A^\dagger) \Bigr) \Bigr|^2
&= \Bigl| \tr \Bigl( \rho_\beta\, \r e^{t\c L/2}(A)\, \r e^{t(i\delta + \c L/2)}(A^\dagger) \Bigr) \Bigr|^2 \\
&\le \Bigl\{ \tr \Bigl( \rho_\beta\, \r e^{t\c L/2}(A)\, \r e^{t\c L/2}(A^\dagger) \Bigr) \Bigr\}^2 \\
&= \Bigl\{ \tr \Bigl( \rho_\beta A\, \r e^{t\c L}(A^\dagger) \Bigr)
\Bigr\}^2.
\end{align} 
Therefore the time auto-correlation functions of observables evolving under the dissipative part of the dynamics alone determine the useful lifetime of observables. As a detailed balance generator $\c L$ is normal, the lifetime is determined by the smallest eigenvalue of $-\c L$ different from 0.

Generally, let $\c H$ be a subspace of observables containing $\idty$ and globally invariant under an ergodic detailed balance generator $\c L$, then
\begin{align}
&\r{Gap}\Bigl( \c G \bigr|_{\c H} \Bigr) 
\nonumber \\
&\quad:= \min\Bigl( \Bigl\{ -\lambda \,:\, 0 \ne \lambda \text{ eigenvalue of } \c L\bigr|_{\c H} \Bigr\} \Bigr)
\label{gap} \\
&\quad = \min\Bigl( \Bigl\{ -\< X \,,\, \c L(X) \>_\beta \,:\, X \in \c H,\ \norm X_\beta = 1, \text{ and } \< \idty \,,\, X \>_\beta = 0 \Bigr\} \Bigr). 
\end{align}

\section{Master Hamiltonians}
\label{s4}

For explicit computations it is convenient to rewrite the Davies generator in the standard Hilbert-Schmidt space. E.g., when dealing with composite systems, the Liouville scalar product introduces correlations between the different parties while the simplicity of the tensor structure is clearly visible in the Hilbert-Schmidt picture. Associating to an observable $X$ the vector
\begin{equation}
\varphi := X\, \rho_\beta^{\frac{1}{2}}, 
\label{uni}
\end{equation}
we pass from Liouville to Hilbert-Schmidt space
\begin{equation} 
\< \varphi\,,\,\psi \> := \tr \varphi^\dagger\psi. 
\end{equation} 
The map~(\ref{uni}) is unitary as $\<X \,,\, X\>_\beta = \<\varphi \,,\, \varphi\>$.

We can now unitarily transport the action of $-\c L_{\alpha\omega}$, see~(\ref{liopos}), on Liouville space to an action $k_{\alpha\omega}$ on Hilbert-Schmidt space
\begin{align}
k_{\alpha\omega} &= \bigl\{ S_\alpha(\omega)_L - \eta\, S_\alpha(\omega)_R \bigr\}^* \bigl\{ S_\alpha(\omega)_L - \eta\, S_\alpha(\omega)_R \bigr\}
\nonumber \\
&\phantom{= \bigl\{} + \bigl\{ S_\alpha(\omega)_R - \eta\, S_\alpha(\omega)_L \bigr\}^* \bigl\{ S_\alpha(\omega)_R - \eta\, S_\alpha(\omega)_L \bigr\}
\label{int} \\
&= \bigl\{ \bigl(S_\alpha(\omega)^\dagger\bigr)_L - \eta\, \bigl(S_\alpha(\omega)^\dagger\bigr)_R \bigr\} \bigl\{
S_\alpha(\omega)_L - \eta\, S_\alpha(\omega)_R \bigr\}
\nonumber \\
&\phantom{= \bigl\{} + \bigl\{ \bigl(S_\alpha(\omega)^\dagger\bigr)_R - \eta\, \bigl(S_\alpha(\omega)^\dagger\bigr)_L \bigr\} \bigl\{ S_\alpha(\omega)_R - \eta\, S_\alpha(\omega)_L \bigr\}.
\end{align}
Here we introduced left and right multiplication 
\begin{equation} 
X_L\, \varphi := X\,\varphi \enskip\text{and}\enskip X_R\, \varphi := \varphi\, X 
\end{equation} 
and 
\begin{equation}
\eta := \exp(-\beta\omega/2).
\label{eq:eta}
\end{equation} 
A $\ast$ was also used instead of a $\dagger$ to distinguish between the Hermitian conjugates for the usual and for the Hilbert-Schmidt inner product.

By~(\ref{int}) each $k_{\alpha\omega}$ is manifestly positive and the same is true for the sum
\begin{equation}
K = \sum_\alpha \sum_{\omega \ge 0} k_{\alpha\omega}.
\label{ham:master}
\end{equation}
The operator $K$ is minus the unitary transform of the Davies generator $\c L$ under the map~(\ref{uni}). In particular $K$ and $-\c L$ have the same eigenvalues, taking multiplicities into account.

We have now reached a Hamiltonian picture of $\c L$ in terms of $K$ which is a sum of positive contributions with the remarkable property that these contributions have a common zero energy vector, the identity operator seen as the vector $\rho_{\beta}^{1/2}$ in Hilbert-Schmidt space. This is rather unusual as the different terms in~(\ref{ham:master}) generally don't commute, implying that $K$ is a truly quantum Hamiltonian. This fictitious Hamiltonian, which will always be denoted by $K$, should be distinguished from the system Hamiltonian $H^{\text{sys}}$, we will call it master Hamiltonian. The positive constants $\hat g$ are not very relevant, we can in fact replace the $\hat g$ by arbitrary positive constants without changing the ground state of $K$. Even more, the full spectrum of $K$ can be sandwiched between spectra of $K$'s with modified constants.

As argued at the end of Section~\ref{s3}, the important parameter characterizing the relaxation time of the system is the spectral gap of the Davies generator $\c L$, see~(\ref{gap}), which is equal to the ground state gap of the master Hamiltonian $K$ as both are unitarily equivalent. Consider a non-negative matrix $A$ with a non-trivial kernel. The smallest strictly positive eigenvalue $g$ of $A$ is its gap above the kernel and will be shortly denoted by $\r{Gap}(A)$. We shall use in the sequel a number of simple facts about gaps.

\begin{lemma}
\label{lem:1}
i) Let $A$ be a positive operator with non-trivial kernel, then any real number $g$ such that $A^2 \ge g A$ is a lower bound for $\r{Gap}(A)$. \\
ii) Let $A$ and $B$ be positive operators such that $\r{Ker}(A+B)$ is non-trivial and that $\r{Ker}(A+B) = \r{Ker}(B)$, then $\r{Gap}(A+B) \ge \r{Gap}(B)$. \\
iii) Let $A$ and $B$ be commuting positive operators such that $\r{Ker}(A+B)$ is non-trivial, then $\r{Gap}(A+B) \ge min\bigl( \{\r{Gap}(A), \r{Gap}(B)\} \bigr)$.
\end{lemma}

We shall also need the following lemma:

\begin{lemma}
\label{lem:2}
For positive operators $A$ and $B$, let $A$ have gap $g_A$ and $\<\varphi \,,\, B\phi\> \ge g_B$ for all normalized $\varphi\in \r{Ker}(A)$, then 
\begin{equation}
A + B \ge \frac{g_A g_B}{g_A + \norm B}.
\label{lem:2:1}
\end{equation}
\end{lemma}

\begin{proof} 
For any normalized vector $\psi$ we write 
\begin{equation}
\psi = a \varphi + b \varphi'
\label{lem:2:2}
\end{equation} 
where $\varphi \in \r{Ker}(A)$, $\varphi'\in \r{Ker}(A)^\perp$ are normalized and $\abs a^2 + \abs b^2 = 1$. We then have 
\begin{equation}
 \< \psi \,,\, (A + B)\psi \> \ge (\eta \,,\, M\eta) \ge \lambda_-,
\label{lem:2:3}
\end{equation}
where $(\cdot\,,\,\cdot)$ is the scalar product in $\Cx^2$, $\eta = (a,b)$,
\begin{equation}
M = \begin{bmatrix} 0&0 \\0&g_A \end{bmatrix} + \begin{bmatrix} \<\varphi \,,\, B \varphi\> &\<\varphi \,,\, B \varphi'\> \\\<\varphi' \,,\, B \varphi\>& \<\varphi'\,,\, B \varphi'\> \end{bmatrix}.
\label{lem:2:4}
\end{equation}
and $\lambda_-$ is the smaller eigenvalue of the positive matrix $M$. The required estimate $\lambda_-$ is obtained from Lemma~\ref{lem:3} below, noting that $\norm B \ge \norm M$.
\end{proof}

\begin{lemma}
\label{lem:3}
Consider the $2\times2$ matrix $C$ given by 
\begin{equation}
\begin{bmatrix} y&x \\ x^*&z+u \end{bmatrix}
\label{lem:3:1}
\end{equation}
where $u>0$ and $C'$ is positive. Here, 
\begin{equation}
C' = \begin{bmatrix} y&x \\ x^*&z \end{bmatrix}.
\label{lem:3:2}
\end{equation} 
Let $\epsilon_-$ be the smaller eigenvalue of $C$, then
\begin{equation}
\epsilon_- \ge \frac{y u}{u + \norm{C'}}.
\label{lem:3:3}
\end{equation}
\end{lemma}

\section{The Ising ferromagnet on a ring}
\label{s5}

The aim of this section is to analyze the relaxation to equilibrium of a system of $N$ Ising spins on a ring weakly coupled to a thermal environment. The model is too simple to be a reasonable candidate for a memory but it will provide useful techniques to deal with the 2D Kitaev model.

\subsection{The model}
\label{s4:1}

The elementary building blocks of the model are two-level systems with the usual Pauli matrices and the identity as a basis of the observables
\begin{equation} 
\sigma^x = \begin{bmatrix} 0 &1 \\ 1 &0 \end{bmatrix},\enskip \sigma^y = \begin{bmatrix} 0 &-i \\ i &0 \end{bmatrix},\enskip \text{and } \sigma^z = \begin{bmatrix} 1 &0 \\ 0 &-1 \end{bmatrix}.
\end{equation} 
There are $N$ such spins arranged in a ring configuration, a subscript $j = 1,2, \ldots, N$ denotes the position of the spin and $N+1$ is identified with 1.

The interaction for an Ising ferromagnetic ring is given by bond observables
\begin{equation} 
Z_b := \sigma_j^z \sigma_{j+1}^z, \enskip\text{with}\enskip b = \{j,j+1\}. 
\end{equation} 
Clearly these bonds are not independent as they satisfy the relation
\begin{equation}
\prod_b Z_b = \idty.
\label{ring} 
\end{equation}
The ferromagnetic Ising Hamiltonian is then
\begin{equation}
H^{\text{Ising}} := - \sideset{}{'}\sum_b J\, Z_b,\enskip J > 0
\label{ising}
\end{equation}
where the prime on the summation reminds that we must take cyclic boundary conditions~(\ref{ring}) into account. This $N$ spin model has a two-fold degenerate ground state. The $Z_b$, and any product of $Z_b$'s, all have expectation 1 which determines the state completely but for one freedom in $\Cx^2$. Precisely this freedom is used to encode a single logical qubit. As before, the term qubit will be used to distinguish between the physical spins and an encoded abstract qubit which is exhibited by the following construction.

The observables
\begin{equation}
\s X := \sigma_1^x \sigma_2^x \cdots \sigma_N^x \enskip\text{and}\enskip \s Z := \sigma_1^z
\label{isingqubit}
\end{equation}
commute with all the bond observables and are therefore constants of the motion for the Hamiltonian evolution. Moreover, they satisfy the qubit relations
\begin{equation}
\s X = \s X^\dagger,\enskip \s Z = \s Z^\dagger,\enskip \s X^2 = \s Z^2 = \idty\enskip \text{and } \s X\, \s Z + \s Z\, \s X = 0.
\label{qubitrel}
\end{equation}
Actually, the Hamiltonian~(\ref{ising}) could be taken more generically without altering the picture. Instead of~(\ref{ising}) we could consider
\begin{equation}
H^{\text{Ising}} := - \sideset{}{'}\sum_b J_b\, Z_b,\enskip J_b > \r{Cst}.
\end{equation} 
We would still have the same ground state and the constants of the motion would now be the algebra generated by the bond observables. The qubits~(\ref{isingqubit}) would still be preserved under this more generic kind of evolution.   

Let $\c Q$ the be the qubit algebra generated by $\s X$ and $\s Z$. We shall decompose the observables with respect to $\c Q$. An admissible bond on a ring of $N$ sites is a configuration $|\b b\> = |b_1, b_2, \ldots, b_N\>$ with $b_j \in \{+,-\}$ containing an even number of minuses. Next consider a Hilbert space $\c H_+$ spanned by an orthonormal family $\{|\b b\>\}$ with $\b b$ admissible. The algebra of linear transformations of $\c H_+$ will be called the full bond algebra and denoted by $\c A^{\text{full}}_{r b}$.

\begin{lemma}
\label{lem:4} The algebra of observables of the Ising model on a ring can be decomposed into the tensor product
\begin{equation}
Q \otimes \c A^{\r{full}}_{\r b}.
\end{equation}
\end{lemma}

\begin{proof}
A natural orthonormal basis of the Hilbert space of the full system consists of the tensor basis $|\bi \epsilon\>$ where $|\bi \epsilon \> = |\epsilon_1, \epsilon_2, \ldots, \epsilon_N\>$ corresponds to the eigenvalues of all $\sigma^z$'s. An alternative labelling is to specify the eigenvalue of the qubit $\s Z$ and $|\b b\>$ of all the admissible bond observables. $\c A^{\text{full}}_{\r b}$ is then generated by the bond observables and by all operators that flip two intersecting bonds and leave $\c Q$ untouched, i.e.\ by $\sigma^x_j$ with $j \neq 2$.
\end{proof}

\subsection{The gap of the generator}
\label{s5:2}

Let us now turn to the ring in a thermal environment. A generic interaction Hamiltonian is of the form
\begin{equation}
H^{\text{int}} = \sum_{j=1}^N \sigma^x_j \otimes f_j + \sum_{j=1}^N \sigma^y_j \otimes f_{N+j} + \sum_{j=1}^N \sigma^z_j \otimes f_{2N+j}. 
\label{isingint}
\end{equation}
where $f_j$ is a self-adjoint field of the $j$-th bath. To each of the terms in~(\ref{isingint}) there corresponds a number of terms in the Davies generator labelled by the Bohr frequencies of the coupling operators. As explained in Section~\ref{s2} every contribution is a positive operator on Liouville space and the identity is a common element of their kernels.

As $H^{\text{sys}}$, see~(\ref{ising}), is a sum of local terms which commute amongst themselves, the Hamiltonian system dynamics inherits this strong locality property: an observable living at site $j$ will at most spread to the sites $j-1$, $j$, and $j+1$. We explicitly compute the evolution of the $\sigma^x_j$ coupling operators
\begin{equation}
\r e^{itH^{\text{Ising}}}\, \sigma^x_j\, \r e^{-itH^{\text{Ising}}} = \r e^{-4iJt}\, a_j + \r e^{4iJt}\, a_j^\dagger + a_j^0.
\label{evol}
\end{equation}
The Fourier components of $\sigma^x_j$ are given in terms of projection operators
\begin{equation} 
P_j^0 = \frac{1}{2}\, \bigl( \idty - Z_b Z_{b'} \bigr) \enskip\text{and}\enskip P_j^{\pm} = \frac{1}{4}\, \bigl( \idty \mp Z_b \bigr)\, \bigl( \idty \mp Z_{b'} \bigr) 
\end{equation} 
with $j = \{b,b'\}$:
\begin{equation}
\begin{split}
&a_j^0 = P_j^0\, \sigma^x_j\, P_j^0 = P_j^0\, \sigma^x_j = \sigma^x_j\, P_j^0 \\
&a_j = P_j^-\, \sigma^x_j\, P_j^+ = P_j^-\, \sigma^x_j =
\sigma^x_j\, P_j^+.
\end{split}
\label{sising}
\end{equation}
Note that the projectors select neighbouring bond configurations that are equal (up or down) or different and that they satisfy $P^0 + P^+ + P^- = \idty$. The Bohr frequencies of the Ising Hamiltonian are $0,\ \pm4J,\ \pm8J,\ldots$ but only the Bohr frequencies $0$ and $\pm 4J$ contribute to the evolution of the coupling operators $\sigma^x_j$, this is due to the strong locality property of the Ising dynamics. Similar expressions hold for the $\sigma^y_j$ terms, it suffices to replace $x$ by $y$ in~(\ref{evol}--\ref{sising}). As the detailed structure of the bath operators $f_j$ is rather irrelevant and to avoid inessential constants we shall choose the $\hat g_\alpha(\omega)$ in the Davies generator~(\ref{dav}) in a convenient way, always assuming that they are strictly larger than a positive constant independent of the system size and the frequency. 
To be specific, the constants are chosen in such a way that the contributions coming from $\sigma^x_j$ (or $\sigma^y_j$) are of the following form:
\begin{equation}
\begin{split}
\c L_j^x(A) = \frac{1}{2} \Bigl\{
&h_+\,a_j^\dagger\, [A \,,\, a_j] + h_+\,[a_j^\dagger \,,\, A]\, a_j + h_-\, a_j\, [A \,,\, a_j^\dagger] \\
&+ h_-\, [a_j \,,\, A]\, a_j^\dagger \Bigr\} - \frac{1}{2}
h_0\,[a^0_j \,,\, [a^0_j \,,\, A]].
\end{split}
\label{eq:gen-L1x}
\end{equation}
\begin{equation}
h_+ = \frac{2}{\gamma^2 + 1},\quad h_- = \frac{2\gamma^2}{\gamma^2 + 1},\quad h_0 = 1,
\label{eq:const-h}
\end{equation}
so that $h_+ + h_- = 2h_0$. Here 
\begin{equation}
\gamma=e^{-2J\beta}.
\label{eq:gamma}
\end{equation}

\begin{theorem}[Gap Ising ring]
\label{thm:1} 
Let $\c L$ be the Davies generator for a system of $N$ Ising spins on a ring, see~(\ref{ising}), obtained by the weak-coupling procedure from the interaction Hamiltonian $H^{\r{int}}$ as in~(\ref{isingint}), then
\begin{equation} 
\r{Gap}(\c L)  \ge \frac{1}{3}\, \r e^{-8\beta J}. 
\end{equation} 
\end{theorem}

\begin{proof}
In Section~\ref{s4} we showed that the gap of the Davies generator coincides with the gap of the master Hamiltonian which consists of a sum of positive contributions with a unique common zero-energy eigenvector $\idty$ in Liouville space. By Lemma~\ref{lem:1}~ii) we will only lower the gap by dropping terms in the master Hamiltonian~(\ref{ham:master}) provided we don't introduce additional zero-energy vectors in doing so. We shall use this property to obtain a simpler expression for a lower bound on the gap. In fact we shall only retain terms coming from $\sigma^x_j$, $j=1,2,\ldots, N$ and $\sigma^y_1$. This leads to the following structure for the simplified generator
\begin{equation}
\c L_1^x + \c L_1^y + \tilde{\c L}. 
\label{simgen}
\end{equation}
Here $\c L_1^x$ refers to the contribution coming from $\sigma^x_1$, $\c L_1^y$ to that from $\sigma_1^y$ and $\tilde{\c L}$ from the contributions of the $\sigma_j^x$ with $j=2,3, \ldots,N$.

Clearly, by~(\ref{isingqubit}) all $\sigma^x_j$ with $j=2,3,\ldots, N$ belong to $\c A^{\text{full}}_{\r b}$ and so do the projectors $P^0_j$ and $P^\pm_j$ for all $j$. Therefore the Davies operators $a_j$ and $a^0_j$ with $j\neq1$ belong to the algebra $\c A^{\text{full}}_{\r b}$. This implies that $\tilde{\c L}$ leaves $\c Q$ untouched and that $\c A^{\text{full}}_{\r b}$ is globally invariant. By abuse of notation we shall from now on use the notation $\id \otimes \tilde{\c L}$ instead of $\tilde{\c L}$.

Next we consider the orthogonal decomposition of $\c Q$ into a direct sum of four 1D spaces, generated by $\idty$, $\s X$, $\s Y$ and $\s Z$ respectively. Here $\s Y = i \s Z \s X$. Note that this is an orthogonal decomposition not only for the Liouville scalar product but for the Hilbert-Schmidt scalar product as well. This decomposition induces one of the full algebra of observables as
\begin{equation}
\Bigl( \idty \otimes \c A^{\text{full}}_{\r b} \Bigr) \oplus \Bigl( \s X \otimes \c A^{\text{full}}_{\r b} \Bigr) \oplus \Bigl( \s Y \otimes \c A^{\text{full}}_{\r b} \Bigr) \oplus \Bigl( \s Z \otimes \c A^{\text{full}}_{\r b} \Bigr). 
\label{ortdec}
\end{equation}
Clearly $\id \otimes \tilde{\c L}$ leaves each of these blocks invariant and has a same action on the $\c A^{\text{full}}_{\r b}$ factors. Moreover, both $\c L_1^x$ and $\c L^y_1$ are block-diagonal in the decomposition~(\ref{ortdec}). This can be seen by considering their explicit actions on elements of the type $A \otimes B$ with $A \in \{\idty, \s X, \s Y, \s Z\}$ and $B \in \c A^{\text{full}}_{\r b}$.

As a first simplification, we observe that each of the three terms in~(\ref{simgen}) is negative on $\s X \otimes \c A^{\text{full}}_{\r b}$, $\s Y \otimes \c A^{\text{full}}_{\r b}$ and $\s Z \otimes \c A^{\text{full}}_{\r b}$. We further simplify our description by retaining only $\c L^y_1$ on $\s X \otimes \c A^{\text{full}}_{\r b}$, $\c L^x_1$ on $\s Y \otimes \c A^{\text{full}}_{\r b}$ and either one of the two on $\s Z \otimes \c A^{\text{full}}_{\r b}$. 

We shall now bound from below $-\tilde{\c L}-\c L_1^x$ restricted to $\s Z \otimes \c A^{\text{full}}_{\r b}$, the two other parts can be handled in the same way. By Proposition~\ref{pro:1}, which will be proved in Section~\ref{s5:3}, $\tilde{\c L}$ is ergodic and gapped with the estimate $h_-/2$ given by ~(\ref{pro:1:1}). Hence $\tilde{\c L}$ restricted to $\s Z \otimes \c A^{\text{full}}_{\r b}$ has gap bounded by $h_-/2$ above its kernel $\s Z\otimes \idty$. Lemma~\ref{lem:2} then says that on this subspace
\begin{equation}
-\tilde{\c L} -\c L_1^x \ge \frac{- \frac{h_-}{2}\, \< \s Z \otimes \idty \,,\, \c L_1^x \s Z \otimes \idty \>}{\frac{h_-}{2} + \norm{\c L_1^x}}.
\end{equation}
Noting that  
\begin{equation}
G := h_+ P_1^+ + h_- P_1^- + h_0P_1^0 \ge h_-
\label{eq:G}
\end{equation}
with $h_\pm,h_0$ given by (\ref{eq:const-h}), 
we obtain 
\begin{equation}
-\<\s Z \otimes \idty \,,\, \c L_1^x\,\s Z \otimes \idty\> = 2\<\idty \,,\, G\> \ge 2 h_-\quad \text{and}\quad \norm{\c L_1^x} \le 2.
\end{equation}
This gives us the following estimate 
\begin{equation}
\bigl( -\tilde{\c L} -\c L_1^x \bigr) \Bigl|_{\s Z \otimes \c A^{\text{full}}_{\r b}} \ge \frac{h_-^2}{\frac{h_-}{2} + 2}.
\end{equation}

For the first term $\idty \otimes \c A^{\text{full}}_{\r b}$ in~(\ref{ortdec}) we retain only $\id \otimes \tilde{\c L}$. As said, $\tilde{\c L}$ is ergodic and gapped with the bound $h_-/2$. We therefore conclude that
\begin{equation} 
\r{Gap}(\c L) \ge  \min\Bigl( \Bigl\{ \frac{h_-}{2}, \frac{h_-^2}{\frac{h_-}{2} + 2} \Bigr\} \Bigr) = \frac{ h_-^2}{\frac{h_-}{2} + 2} \ge \frac{h_-^2}{3} \ge \frac{\gamma^4}{3}
\end{equation} 
since $1 \ge h_- \ge \gamma^2$.
\end{proof}

\subsection{A bond model}
\label{s5:3}

In this section we consider the spectral gap of the generator $\tilde{\c L}$ introduced in the proof of Theorem~\ref{thm:1}. As before, let $\c H_+$ be the Hilbert space spanned by the basis vectors $|\b b\>$ labelled by bonds on a ring, i.e.\ $|\b b\> = |b_1, b_2, \ldots, b_N\>$ with $b_j \in \{+,-\}$ and $|\b b\>$ contains an even number of minuses. Consider the following operators on $\c H^+$
\begin{equation}
\begin{split}
&a_j = (|--\>\<++|)_{bb'},\enskip a^\dagger_j = (|++\>\<--|)_{bb'},\enskip\text{and} \\
&a^0_j = (|-+\>\<+-|)_{bb'} + (|+-\>\<-+|)_{bb'}
\end{split}
\label{part}
\end{equation}
where the bonds $b$ and $b'$ satisfy $j = \{b,b'\}$. We are interested in the spectral gap of the operator $\tilde{\c L}$ acting on the linear transformations $\c B(\c H^+)$ of $\c H^+$
\begin{equation}
\begin{split}
\tilde{\c L}(X) = \frac{1}{2} \sum_{j=2}^N \Bigl\{
&h_+\, a_j^\dagger\, [X \,,\, a_j] + h_+\, [a_j^\dagger \,,\, X]\, a_j + h_-\,  a_j\, [X \,,\, a_j^\dagger] \\
&+ h_-\, [a_j \,,\, X]\, a_j^\dagger - h_0\,[a^0_j \,,\, [a^0_j \,,\, X]] \Bigr\}.
\end{split}
\label{eq:gen-L}
\end{equation}

\begin{proposition}
\label{pro:1} 
The generator $\tilde{\c L}$ is ergodic and gapped with bound
\begin{equation} 
\r{Gap}(\tilde{\c L}) \ge \frac{h_-}{2} = \frac{\gamma^2}{1 + \gamma^2}
\label{pro:1:1}
\end{equation} 
and $\gamma$ defined in~(\ref{eq:gamma}).
\end{proposition}

In order to estimate the gap of $\tilde{\c L}$ we shall use details on its structure as an operator on Liouville or Hilbert-Schmidt space. For any $\Lambda \subset \{1,2,\ldots,N\}$ with $\bigl( N - \#(\Lambda) \bigr)$ even we consider the subspace $\c B(\c H_+)(\Lambda)$ of $\c B(\c H_+)$ spanned by the rank one operators $|\b b\>\<\b{b'}|$ where $\b b$ and $\b{b'}$ coincide on $\Lambda$ and where all the bonds of $\b{b'}$ sitting on the complement of $\Lambda$ are flipped with respect to these of $\b b$. These subspaces are orthogonal, both with respect to the Liouville and the Hilbert-Schmidt scalar products, so that we can write
\begin{equation}
\c B(\c H_+) = \loplus_\Lambda \c B(\c H_+)(\Lambda). 
\label{blocks}
\end{equation}
It is furthermore not hard to see that
\begin{equation} 
\c B(\c H_+)(\Lambda) = \c A^{\text{ab}}(\Lambda) \otimes \c F(\Lambda). 
\end{equation} 
Here $\c A^{\text{ab}}(\Lambda)$ is the Abelian algebra generated by the projectors on bonds restricted to $\Lambda$ and $\c F(\Lambda)$ is the space spanned by proper flips of bonds restricted to the complement of $\Lambda$. 
It is straightforward to check that every block $\c B(\c H_+)(\Lambda)$ is left invariant by $\tilde{\c L}$. Due to the positivity of $\tilde{\c L}$ we obtain the following structure for the generator.

\begin{lemma}
\label{lem:5} 
The operator $\tilde{\c L}$, see~(\ref{eq:gen-L}) restricted to the algebra  $\c B(\c H_+)$ is block diagonal with respect to the decomposition~(\ref{blocks}).
\end{lemma}

\begin{proof}[Proof of Proposition~\ref{pro:1}]
We rely on the block diagonal structure of $\tilde{\c L}$ as described in Lemma~\ref{lem:5}. The proof consists of the following steps \\
1) obtaining a lower bound for minus the restriction of $\tilde{\c L}$ to blocks with non-trivial flip part, which follows immediately this enumeration \\
2) obtaining a lower bound for the gap of the generator $\tilde{\c L}_{\text{ab}}$ which is the restriction of $\c L$ to the subspace with trivial flip part, see Proposition~\ref{pro:2} \\
3) proving the ergodicity of $\tilde{\c L}$, see Lemma~\ref{lem:6}.

A partition of $\{1,2,\ldots, N\}$ into $\Lambda$ and its complement induces a partition of the spins $j = \{b,b'\}$ into three sets: $\Gamma_{\text{flip}}$, $\Gamma_{\text{ab}}$ and $\Gamma_{\text{int}}$, consisting of pairs of neighbouring bonds from $\Lambda^c$, pairs of bonds from $\Lambda$ and pairs that have one bond in $\Lambda$ and the other in its complement. This partition induces accordingly a decomposition of the generator $\tilde{\c L}$. By explicit use of~(\ref{part}) one verifies that for $j \in \Gamma_{\text{ab}}$ the flip part is left untouched while the Abelian part is mapped into itself. The same holds true for $j \in \Gamma_{\text{flip}}$. It remains to consider contributions coming from $j \in \Gamma_{\text{int}}$. We just give an explicit example. Suppose that $j = \{b,b'\}$ with $b\in\Lambda$ and $b'\in\Lambda^c$ such as $(|++\>\<+-|)_{bb'}$. Using~(\ref{part}) we find
\begin{equation} 
\tilde{\c L}_j(|++\>\<+-|) = - \frac12(h_- + h_0) |++\>\<+-|.
\end{equation} 
with $h_-,h_0$ given by (\ref{eq:const-h}).
It is then straightforward to verify that
\begin{equation} 
-\tilde{\c L} \ge \#(\Gamma_{\text{int}}) h_-\geq h_-.
\label{eq:Gamma-int}
\end{equation} 
We should therefore distinguish three cases: the Abelian case $\Gamma_{\text{flip}} = \emptyset$, see Proposition~\ref{pro:2}, the pure flip case $\Gamma_{\text{ab}} = \emptyset$, see the remarks following the proof of Proposition~\ref{pro:2}, and the case $\Gamma_{\text{int}} \neq \emptyset$ which was dealt with above.
\end{proof}

Let $\tilde{\c L}_{\text{ab}}$ be the restriction of the generator $\tilde{\c L}$ to the diagonal algebra generated by projectors of the form $|\b b\>\<\b b|$ with $\b b$ an admissible bond on the ring. As the Gibbs density matrix $\rho_\beta$ is also diagonal these projectors are an orthogonal family. Minus the restricted generator $\tilde{\c L}_{\text{ab}}$ can now easily be expressed with respect to the Hilbert-Schmidt inner product
\begin{equation} 
\tilde K_{\text{ab}} = \sum_{bb'} k_{bb'} 
\end{equation} 
where
\begin{equation} 
k_{bb'} = \begin{pmatrix} \frac{\gamma^2}{1+\gamma^2}&0&0&-\frac{\gamma}{1+\gamma^2}\\0&\frac12&-\frac12&0\\ 0&-\frac12&\frac12&0 \\-\frac{\gamma}{1+\gamma^2}&0&0&\frac{1}{1+\gamma^2} \end{pmatrix}. 
\end{equation} 
Due to our choice of constants $g_\alpha(\omega)$ (as done in \ref{eq:gen-L1x}) are projectors. The summation extends over an open chain of bonds, indeed the contribution of the site $j=1$ of the ring was removed.

\begin{proposition}[Gap on Abelian factor]
\label{pro:2} 
With the notation of above
\begin{equation}
\r{Gap}\bigl( \tilde{\c L}_{\r{ab}}\bigr) = \r{Gap}\bigl( \tilde K_{\r{ab}} \bigr) \ge \frac{\gamma^2}{1+\gamma^2}. 
\label{pro:2:1}
\end{equation}
\end{proposition}

\begin{proof}
A lower bound for the gap of $\tilde K_{\text{ab}}$ can be obtained by comparing $\bigl(\tilde K_{\text{ab}}\bigr)^2$ with a multiple of $\tilde K_{\text{ab}}$, see Lemma~\ref{lem:1}~i). We start out by writing
\begin{align} 
\bigl( \tilde K_{\text{ab}} \bigr)^2 
&= \, \Bigl( k_{12} + k_{23} + \cdots + k_{N-1\,N} \Bigr)^2 \\
&=  k_{12}^2 + k_{23}^2 + \cdots + k_{N-1\,N}^2 + k_{12} k_{23} + k_{23} k_{12} + \cdots \\
&=  k_{12} + k_{23} + \cdots + k_{N-1\,N} + k_{12} k_{23} + k_{23} k_{12} + \cdots 
\end{align} 
where we used that $k_{bb'}$ is a projector. Suppose now that we find a positive $g$ such that
\begin{equation}
k_{12} + k_{23} + 2 k_{12} k_{23} + 2 k_{23} k_{12} \ge g \bigl( k_{12} + k_{23} \bigr), 
\label{pro:2:3}
\end{equation}
then we may obtain a lower bound for $\tilde K_{\text{ab}}^2$. Projectors $k$ that have no bond in common commute, implying that products of such projectors are positive. We may also decrease the coefficient of a projector. This leads to
\begin{align} 
\bigl( \tilde K_{\text{ab}} \bigr)^2 
&\ge \, k_{12} + k_{23} + \cdots + k_{12} k_{23} + k_{23} k_{12} + \cdots \\
&\ge \,  \Bigl( \frac{1}{2}\, k_{12} + \frac{1}{2}\, k_{23} + k_{12} k_{23} + k_{23} k_{12} \Bigr) \nonumber\\
&\phantom{\ge \, \Bigl\{\ }+ \Bigl( \frac{1}{2}\, k_{23} + \frac{1}{2}\, k_{34} + k_{23} k_{34} + k_{34} k_{23} \Bigr) + \cdots \nonumber \\
&\phantom{\ge \, \Bigl\{\ }+ \Bigl( \frac{1}{2}\, k_{N-2\, N-1} + \frac{1}{2}\, k_{N-1\, N} + k_{N-2\,N-1} k_{N-1\,N} + k_{N-1\,N} k_{N-2\,N} \Bigr) \\
&\ge g\, \Bigl\{ \Bigl( \frac{1}{2}\, k_{12} + \frac{1}{2}\, k_{23} \Bigr) + \Bigl( \frac{1}{2}\, k_{23} + \frac{1}{2}\, k_{34} \Bigr) + \cdots \nonumber \\
&\phantom{\ge \, \Bigl\{\ }+ \Bigl( \frac{1}{2}\, k_{N-2\, N-1} + \frac{1}{2}\, k_{N-1\, N} \Bigr)\Bigr\} \\
&\ge \frac{g}{2}\, \tilde K_{\text{ab}}. 
\end{align} 
It remains to obtain a $g$ satisfying~(\ref{pro:2:3}). As the $k$'s are projectors, the condition is equivalent to
\begin{equation} 
\bigl( k_{12} + k_{23} \bigr)^2 \ge \frac{1+g}{2}\, \bigl( k_{12} + k_{23} \bigr). 
\end{equation} 
The eigenvalues of $k_{12} + k_{23}$ are $\{0,2,\frac{3+\gamma^2}{2(1+\gamma^2)}, \frac{1+3\gamma^2}{2(1+\gamma^2)} \}$. This leads to the estimate
\begin{equation} 
\r{Gap}\bigl(\tilde{\c L}_{\text{ab}}\bigr) \ge \frac{\gamma^2}{1+\gamma^2}. 
\end{equation} 
\end{proof}

It remains to investigate the case $\Gamma_{\text{ab}} = \emptyset$, i.e.\ to consider the action of $\tilde{\c L}$ on the space of linear combinations of matrix units of the form $|\b b\>\<\b b'|$ where $\b b$ is an admissible configuration of bonds and where $\b b'$ is the bond configuration obtained from $\b b$ by flipping every bond. The configuration $\b b'$ should also be admissible which is impossible for odd $N$. For even $N$ essentially a similar approach to the one in Proposition~\ref{pro:2} can be used to obtain that $-\tilde {\c L}\geq 1/2$ on this subspace. In order to keep the proof concise we shall assume in the sequel that $N$ is odd. Finally, to finish the proof of Proposition~\ref{pro:1}, it remains to show the ergodicity of $\tilde{\c L}$.

\begin{lemma}[Ergodicity]
\label{lem:6} 
The generator $\tilde{\c L}$ is ergodic on $\c B(\c H_+)$.
\end{lemma}

\begin{proof}
As in the analysis of the gap of the generator, we shall assume that $N$ is odd, this simplifies the arguments. The general case can however be handled as well. The commutant of the Davies operators $a^0_j$ and $a^\pm_j$, $j=2,3,\ldots,N$ is the same as the commutant of $\bigl\{ \sigma^x_j, Z_{b}+ Z_{b'}  \text{ with } j = \{b,b'\} \bigr\}$. As $N$ is odd this turns out to be the commutant of $\c B(\c H_+)$ which is trivial.
\end{proof}

\section{Kitaev's 2D model}
\label{s6}

Even if Kitaev's 2D model is more complicated than the Ising model on a chain many ideas in the analysis of it's thermal stability are quite similar, certainly on the mathematical level. We shall stress this similarity by following the approach used for the Ising case.

\subsection{The model}
\label{s6:1}

Kitaev's 2D model lives on a $L \times L$ toroidal lattice $\Lambda$. The microscopic constituents are spins located at the midpoint of the edges of the lattice and represented by filled and open dots in Fig.~\ref{fig:1}, so there are $2 L^2$ of them. The interactions between the spins are given by \emph{star} and \emph{plaquette} terms $X_s$ and $Z_p$. A star is a cross whose vertical vertices lie on the lattice of black dots, while the vertical vertices of a plaquette lie on the lattice of white dots. The actual star and plaquette observables are
\begin{equation}
X_s = \prod_{j\in s} \sigma_j^x \enskip\text{and}\enskip Z_p = \prod_{j\in p} \sigma_j^z. 
\label{s&p}
\end{equation}
Stars and plaquettes correspond to the grey shapes in the figure. As stars and plaquettes have either 0 or 2 sites in common, $[X_s, Z_p] = 0$. So, the algebras $\c A_{\r X}$, $\c A_{\r Z}$, and $\c A_{\r{XZ}}$ generated by the $X_s$, the $Z_p$, and by both together are Abelian.

The Hamiltonian of the model is
\begin{equation}
H^{\text{Kit}}_\Lambda = - \sum_s J\, X_s  - \sum_p J\, Z_p,\enskip J > 0. 
\label{kit}
\end{equation}
Similarly to the Ising model the ground states are totally unfrustrated: all $X_s$ and $Z_p$ have expectation 1. This is actually not sufficient to fully determine the state of all spins as the star and plaquette observables are not independent: because of the periodic boundary conditions they satisfy
\begin{equation} 
\prod_s X_s = \idty \enskip\text{and}\enskip \prod_p Z_p = \idty.
\end{equation} 
As a consequence two topological qubit freedoms are left which may be used for encoding. The Hamiltonian~(\ref{kit}) can be chosen more generically by multiplying the individual star and plaquette observables by positive but otherwise arbitrary coefficients, this will not change the set of ground states. Here too, it is natural to consider the commutant of such a generic Hamiltonian which consists of a product of two qubit algebras and $\c A_{\r{XZ}}$. This is seen quite explicitly by introducing, similarly to~(\ref{isingqubit}), observables for two encoded qubits
\begin{equation}
\begin{split}
&\s X_1 = \prod_{j\in c_1} \sigma_{j'}^x, \qquad \s X_2 = \prod_{j\in c_2} \sigma_{j'}^x \\ 
&\s Z_1 = \prod_{j\in d_1} \sigma_j^z,\qquad \s Z_2 = \prod_{j\in d_2} \sigma_j^z.
\end{split}
\label{kitaevqubits}
\end{equation}
Here $c_1,\ d_1,\ c_2$ and $d_2$ are the loops shown in Fig.~\ref{fig:1}. Unlike for the Ising ring, all qubit observables are very delocalized.

\begin{figure}
\begin{center}
\includegraphics[width=.45\textwidth]{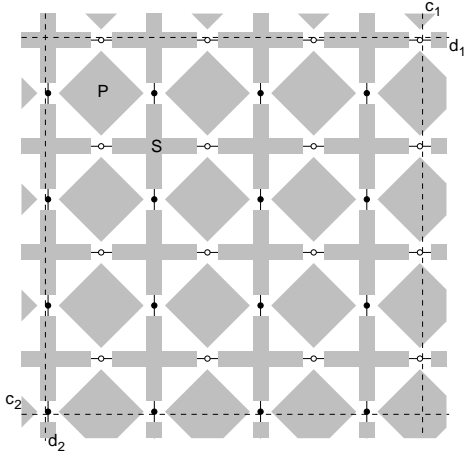}
\caption{Kitaev's lattice\label{fig:1}}
\end{center}
\end{figure}

Let us divide the set of spins into four disjoint subsets, see Fig.~\ref{fig:3}: the snake, the comb, spin 1 and spin 2. Spin 1 is located at the crossing of $\s X_1$ and $\s Z_1$ and similarly for spin 2. Note that the qubit $\s X_1$ has been modified a little, so that it closely follows the snake. 

Let $\c A^{\text{full}}_{\r p}$ be the algebra is generated by all plaquette observables and by the $\sigma^x_j$ where $j$ belongs to the snake while $\c A^{\text{full}}_{\r s}$ is generated by the star observables and by the $\sigma^z_j$ with $j$ belonging to the comb. It is obvious from the construction that these algebras commute and that they are both isomorphic to a full matrix algebra of dimension $2^{L^2-1}$. Indeed, considering $\c A^{\text{full}}_{\r p}$ first we have $L^2-1$ independent commuting plaquette observables. Admissible flips of these plaquette observables should comply with condition~(\ref{s&p}). Now any such flip can be obtained by concatenating flips induced by $\sigma^x_j$ with $j$ belonging to the snake. A similar reasoning holds for $\c A^{\text{full}}_{\r s}$. Therefore the commutant of these two algebras is a 4 dimensional matrix which may be identified with the tensor product of two qubit algebras $\c Q_1$ and $\c Q_2$ generated by the logical qubits~(\ref{kitaevqubits}). We have therefore shown that

\begin{figure}
\begin{center}
\includegraphics[width=.45\textwidth]{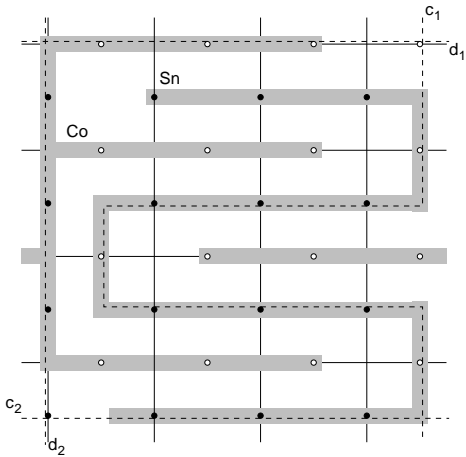}
\caption{Partitioning Kitaev's lattice\label{fig:3}}
\end{center}
\end{figure}

\begin{lemma}
\label{lem:7} 
The algebra of observables of the 2D Kitaev model on a toric $L\times L$ lattice can be decomposed into the tensor product
\begin{equation} 
\c Q_1 \otimes \c Q_2 \otimes \c A^{\r{full}}_{\r p} \otimes \c A^{\r{full}}_{\r s}.
\label{eq:kitaev-algebra}
\end{equation} 
\end{lemma}

There is a striking analogy between Lemma~\ref{lem:7} and Lemma~\ref{lem:4}. The plaquettes or stars should be compared to the bonds. Moreover, the operators that flip plaquettes or stars are located on essentially one-dimensional structures: the snake and the comb. This makes the 2D Kitaev model essentially isomorphic to two Ising like models on 1D structures. For the factors, $\c A^{\text{full}}_{\text{p}}$ and $\c A^{\text{full}}_{\text{s}}$ we shall therefore consider the same block decomposition as for $\c A^{\text{full}}_{\text{b}}$:
\begin{equation}
\c A^{\r{full}}_{\r p} = \loplus_\Lambda \c A^{\r{ab}}_{\r p}(\Lambda) \otimes \c F_{\r p}(\Lambda) \enskip\text{and}\enskip \c A^{\r{full}}_{\r s} = \loplus_\Lambda \c A^{\r{ab}}_{\r s}(\Lambda) \otimes \c F_{\r s}(\Lambda).
\label{eq:block-kitaev}
\end{equation}
Moreover, the factor $\c A^{\text{full}}_{\text{p}}$ is fully isomorphic to the full bond algebra $\c A^{\text{full}}_{\text{b}}$ of the 1D Ising model, hence we shall further on use the same notation, keeping in mind that bonds translate into plaquettes. 

\section{Gap for Kitaev's model}

In the Davies generator we now retain only contributions coming from $\sigma^x_j$ and $\sigma^z_j$. 
The generator has the form 
\begin{equation} 
\c L = \sum_{j=1}^N \c L_{x,j} + \sum_{j=1}^N \c L_{z,j} \equiv \c L_x + \c L_z.
\end{equation} 
where $j$ runs over all qubits.  Note that the action of $\c L$ is ergodic and that $\c L_x$ and $\c L_z$ commute. Therefore it is enough to show that e.g.\ $\c L_x$ has a gap. 

We first note that $\c L_x$ is block diagonal, where the blocks are even finer than those determined by formulas~(\ref{eq:block-kitaev}). Namely, we consider for some fixed $\Lambda$ blocks of the form $\c A_{\r s} \otimes F_{\r s} \otimes \c A^{\r{full}}_{\r p} \otimes \s Z_\nu \s X_\mu$, where $F_s$ is a proper flip from $\c F_{\r s}$. One then finds, similarly to formulas (91)--(94) in~\cite{Alicki-Fannes-Horodecki}), that
\begin{equation}
\c L_x \bigl( \c A_{\r s} \otimes F_{\r s} \otimes \c A^{\r{full}}_{\r p} \otimes \s Z_\nu \s X_\mu \bigr) = \c A_{\r s} \otimes F_{\r s} \otimes \c L'_x \bigl( \c A^{\r{full}}_{\r p} \bigr) \otimes \s Z_\nu \s X_\mu \bigr.
\label{eq:signflip}
\end{equation} 
Here $\c L'_x$ is defined as follows: for any $j$ which belongs to the symmetric difference of $S$, i.e.\ the set of qubits where $F_{\r s}$ has $\sigma_z^j$ matrices, and the support of $Z_\nu$ the term $\c L'_{x,j}(A)$ has the form $-\frac{1}{2}\, (a^\dagger a A +  A a^\dagger a) - a^\dagger A a$. This is so because for such a $j$, the $\sigma_x^j$ term from the generator has to be commuted over a $\sigma_z^j$ occurring in $F_{\r s}$ or in $\s Z_\nu$. This results in flipping the sign
of the contribution $ a^\dagger (A) a$. For all other $j$'s, the generator is unchanged: $\c L'_{x,j} = \c L_{x,j}$.

We now have to show that $\c L_x$ has a gap on each block. To this end we introduce $\c L_x^{\r{snake}}$  which consists solely of contributions from sites belonging to the snake, see Fig.~\ref{fig:3}. We shall use that 
\begin{equation}
\c L_x^{\r{snake}} \bigl( \c A_{\r s} \otimes F_{\r s} \otimes \c A^{\r{full}}_{\r p} \otimes \s Z_\nu \s X_\mu \bigr) = \c A_{\r s} \otimes F_{\r s} \otimes \c L_x^{\r{snake}} \bigl( \c A^{\r{full}}_{\r p} \bigr) \otimes \s Z_\nu \s X_\mu 
\end{equation}
and that 
\begin{equation} 
\r{Gap}\bigl( \c L_x^{\r{snake}} \bigr) \Bigr|_{\c A^{\r{full}}_{\r p}} \ge \frac{h_-}{2} \quad\text{and}\quad \r{Ker} \bigl( \c L_x^{\r{snake}} \bigr) \Bigr|_{\c A^{\r{full}}_{\r p}} = \{ \idty \}.
\end{equation}
Here $h_-$ is the same constant as in~(\ref{eq:const-h}). This follows from Theorem~\ref{thm:1} since $\c L_x^{\r{snake}}$ restricted to $\c A^{\r{full}}_{\r p}$ is fully isomorphic the Ising model on $\c A^{\r{full}}_{\r b}$. 

We shall now consider three types of blocks. For the first two we shall bound $-\c L_x$ from below, 
exactly like in the Ising model for blocks of type $\s Z \otimes \c A^{\r{full}}_{\r b}$. In the third type of block $-\c L_x$ will be gapped by the same estimate as in~(\ref{pro:2:1}).

\textit{Case i).} Consider all blocks where $\s Z_1$ is present, the case where $\s Z_2$ is present is analogous. We then remove lot of interactions and retain
\begin{equation}
\c L''_x = \c L_x^{\r{snake}} + \c L_{q_1}^x.
\end{equation}
Note that $\c L''_x$ has support outside the comb. We have therefore
\begin{equation}
\begin{split}
\c L''_x \bigl( \c A_{\r s} \otimes F_{\r s} \otimes \c A^{\r{full}}_{\r p} \otimes \s Z_1 \s X_\mu \bigr) = 
& \c A_{\r s} \otimes F_{\r s} \otimes \c L^{\r{snake}}_x \bigl( \c A^{\r{full}}_{\r p} \bigr) \otimes \s Z_1 \s X_\mu \\
&+ \c A_{\r s} \otimes F_{\r s} \otimes \c L_{q_1} \bigl( \c A^{\r{full}}_{\r p} \otimes \s Z_1 \bigr) \s X_\mu .
\end{split}
\end{equation}
We now apply Lemma~\ref{lem:2}. From discussion below~(\ref{eq:signflip}) it follows that 
\begin{equation}
\< \s Z_1 \,,\, \c L_{q_1}(\s Z_1)\> \ge 2 \< \idty \,,\, G \>_\beta \ge 2 h_- 
\end{equation}
with $G$ as in~(\ref{eq:G}). Within the block under consideration we obtain exactly the same bound as for the Ising model in Theorem~\ref{thm:1}
\begin{equation}
-\c L''_x \ge \frac{\gamma^4}{3}
\end{equation}
where we have used $\norm{L_{q_1}} \le 2$. This implies the same estimate for the full generator $\c L_x$ on this block 
\begin{equation}
-\c L_x \ge \frac{\gamma^4}{3}.
\label{lbl}
\end{equation}

\textit{Case ii).} Consider now blocks where $Z_\nu$ is not present but $F_s$ is. We then take 
\begin{equation}
\c L''_x = \c L^{\r{snake}}_x + \c L_{x,q_0}
\label{eq:caseii}
\end{equation}
where $q_0$ is any spin belonging to the set $S$ determined by the flip $F_s$. A similar argument as for Case~i) then yields the same estimate~(\ref{lbl}).

\textit{Case iii).} Finally, consider blocks where neither $F_s$ nor $Z_\nu$ is present. We then simply take $\c L^{\r{snake}}_x$ itself. For this kind of block $\c L^{\r{snake}}_x$ has a gap $\frac{h_-}{2}$ above its kernel which is generated by $\c A_{\r s} \otimes X_\mu$. One immediately sees that this coincides with the kernel of $\c L_x$ on this block. Therefore by Lemma~\ref{lem:1} (ii) we have 
\begin{equation}
\r{Gap} \bigl( \c L_x \bigr) \Bigr|_{\c A_{\r s} \otimes \c A^{\r{full}}_{\r p} \otimes \s X_\mu} \ge \frac{h_-}{2}.
\end{equation}

Combining the three cases, we conclude that $\c L_x$ has a gap above its kernel which satisfies 
\begin{equation}
\r{Gap} \bigl( \c L_x \bigr) \ge \frac{\gamma^4}{3}.
\end{equation}
We have therefore completed the proof of the following theorem:

\begin{theorem}[Gap Kitaev's 2D model]
\label{thm:2} 
Let $\c L$ be the Davies generator for the 2D Kitaev model of size $L\times L$, then
\begin{equation} 
\r{Gap}(\c L) \ge \frac{1}{3}\, \r e^{-8\beta J}.
\end{equation} 
\end{theorem}
\section{Conclusion}

We presented in this paper a rigorous analysis confirming the heuristic arguments in~\cite{DennisKLP2002-4DKitaev,AlickiH2006-drive} that the 2D Kitaev model is not stable with respect to thermal fluctuations. The physical mechanism behind this instability is the absence of interactions between the thermal excitations. Indeed, Kitaev's 2D model can be translated into a non interacting quasi-particle model (anyons), similarly to non-interacting kinks in the 1D Ising model. Therefore, except for a gas of bosons where the statistics plays a major role, only systems with truly interacting thermal particles like the 2D Ising or the 3D and 4D  Kitaev models could support delocalized metastable observables. In this paper the analysis of the dissipative generator is translated into the analysis of the ground state of a quantum spin Hamiltonian $K$ called master Hamiltonian. Quantum spin chain techniques can then be used to obtain lower bounds on the spectral gap of $K$, leading to instability results as in this paper. 

The same technique can however also be applied to obtain for suitable models upper bounds for the spectral gap of $K$ and to unravel the form of possible metastable observables. Such models and their relevance for the implementation of quantum memory is the topic of a forthcoming publication.

\noindent
\textbf{Acknowledgements} \\
This work is partly supported by the EU Project QAP-IST contract 015848 (RA), the Polish-Flemish bilateral grant BIL~05/11 (RA \& MF), the Belgian Interuniversity Attraction Poles Programme P6/02 (MF), and EC IP SCALA (MH). Part of the work was done while two of the authors (RA and MF) participated to the 2008 MHQP programme at the IMS, NUS, Singapore, they are both grateful for the stimulating atmosphere.


\begin{thebibliography}{10}

\bibitem{KitaevFT}
Kitaev~A.~Y.
2003
Fault-tolerant quantum computation by anyons
\emph{Annals Phys.}~\textbf{303} 2--30
\texttt{quant-ph/9707021}

\bibitem{DennisKLP2002-4DKitaev}
Dennis~E., Kitaev~A.~Y., Landahl~A., and Preskill~J.
2002
Topological quantum memory
\emph{J.\ Math.\ Phys.}~\textbf{43} 4452--4505
\texttt{quant-ph/0110143}

\bibitem{AlickiH2006-drive}
Alicki~R. and Horodecki~M.
2006
Can one build a quantum hard drive? A no-go theorem for storing quantum information in equilibrium systems
\texttt{quant-ph/0603260}

\bibitem{Alicki-Fannes-Horodecki}
Alicki~R., Fannes~M., and Horodecki~M.
2007 
A statistical mechanics view on Kitaev's proposal for quantum memories
\emph{J. Phys. A: Math. Theor.}~\textbf{40} 6451-6467  
\texttt{quant-ph/0702102v1}

\bibitem{NussinovO2007b}
Nussinov~Z. and Ortiz~G. 2007
Autocorrelations and thermal fragility of anyonic loops in topologically quantum ordered systems
\texttt{arXiv:0709.2717}

\bibitem{Kay2008-K2D}
Kay~A.
2008 
The non-equilibrium reliability of quantum memories
\texttt{arXiv:0807.0287}

\bibitem{Davies1974-WCL}
Davies~E.~B. 
1974
Markovian master equations
\emph{Commun.\ Math.\ Phys.}~\textbf{39} 91--110

\bibitem{Alicki-Lendi}
Alicki~R. and Lendi~K. 
2007
\emph{Quantum Dynamical Semigroups and Applications, II edition}
(Berlin: Springer)

\bibitem{Spohn1977-ergod}
Spohn~H.
1977
An algebraic condition for the approach to equilibrium of an open $N$-level
system
\emph{Lett.\ Math.\ Phys.}~\textbf{2} 33--38

\bibitem{Frigerio1978-ergod}
Frigerio~A. 
1977
Stationary states of quantum dynamical semigroups
\emph{Commun.\ Math.\ Phys.}~\textbf{63} 269--276

\bibitem{Alicki-Fannes-fidelity}
Alicki~R. and Fannes~M. 
2008
Decay of fidelity in terms of correlation functions 
\texttt{arXiv:0809.4180}

\end{thebibliography}
\end{document}